\documentclass[12pt,a4paper]{article}
\usepackage{graphicx}
\usepackage{amsmath}
\usepackage{amsthm}
\usepackage{amssymb}
\newtheorem{theorem}{Theorem}

\newtheorem{lemma}{Lemma}

\author{Alexander Ageev}

\title{Approximating the $2$-Machine Flow Shop Problem with Exact Delays Taking Two Values}

\begin{document}
\date{}
\maketitle

\begin{abstract} In the $2$-Machine Flow Shop problem
with exact delays the operations of each job are separated by a
given time lag (delay). Leung et al. (2007) established that the
problem is strongly NP-hard when the delays may have at most two
different values. We present further results for this case: we prove
that the existence of $(1.25-\varepsilon)$-approximation implies
P$=$NP and develop a $2$-approximation algorithm.
\end{abstract}

\section{Introduction}
\label{intro} An instance of the $2$-Machine Flow Shop problem with
exact delays consists of $n$ triples $(a_j, l_j, b_j)$ of
nonnegative integers where $j$ is a job in the set of jobs
$J=\{1,\dots , n\}$. Each job $j$ must be processed first on machine
$1$ and then on machine $2$, $a_j$ and $b_j$ are the lengths of
operations on machines $1$ and $2$, respectively. The operation of
job $j$ on machine~2 must start exactly $l_j$ time units after the
operation on machine~1 has been completed. The goal is to minimize
makespan. In the standard three-field notation scheme the problem is
written as $F2\;|\mbox{ exact }l_j|\;C_{\max}$.

One  of evident applications of scheduling problems with exact
delays is chemistry manufacturing where there often may be an exact
technological delay between the completion time of some operation
and the initial time of the next operation. The problems with exact
delays also arise in com\-mand-and-control applications in which a
centralized commander distributes a set of orders (associated with
the first operations) and must wait to receive responses
(corresponding to the second operations) that do not conflict with
any other (for more extensive discussion on the subject, see
\cite{ESS,SS}). Condotta \cite{Con} describes an application related
to booking appointments of chemotherapy treatments \cite{Con}.

The approximability of $F2\;|\mbox{ exact }l_j \;|\;C_{\max}$ was
studied by Ageev and Kononov in \cite{AK}. They proved that the
existence of $(1.5-\varepsilon)$-approxi\-mation algorithm implies
P$=$NP and constructed a $3$-approxi\-mation algorithm. They also
give a $2$-approximation algorithm for the cases when $a_j\leq b_j$
and $a_j\geq b_j$, $j\in J$. These algorithms were independently
developed by Leung et al. in \cite{LLZ}. The case of unit processing
times ($a_j=b_j=1$ for all $j\in J$) was shown to be strongly
NP-hard by Yu \cite{Yu,YHL}. Ageev and Baburin \cite{AB} gave a
$1.5$-approximation algorithm for solving this case.

In this paper we consider the case when $l_j\in \{L_1,L_2\}$ for all
$j\in \{1,\ldots , n\}$. In the three-field notation scheme this
case can be written as $F2\;|\mbox{ exact }l_j\in
\{L_1,L_2\}\;|\;C_{\max}$. The problem was shown to be strongly
NP-hard by Leung et al. \cite{LLZ}.

Our results are the following: we prove that the existence of
$(1.25-\varepsilon)$-approximation for $F2\;|\mbox{ exact }l_j\in
\{0,L\}\;|\;C_{\max}$ implies P$=$NP and present a $2$-approximation
algorithm for $F2\;|\mbox{ exact }l_j\in \{L_1,L_2\}\;|\;C_{\max}$.

Throughout the paper we use the standard notation:
$C_{\max}(\sigma)$ stands for the length of a schedule $\sigma$
(makespan); $C_{\max}^*$ means the length of a shortest schedule.

We assume that no job has a missing operation in the sense that a
zero processing time implies that the job has to visit the machine
for an infinitesimal amount of time $\delta > 0$.

\section{Inapproximability lower bound for \\ $F2\;|\mbox{ exact }l_j\in \{0,L\}\;|\;C_{\max}$}

In this section we establish the inapproximability lower bound for
the case $F2\;|\mbox{ exact }l_j\in \{0,L\}\;|\;C_{\max}$, i.e.,
when $L_1=0$, $L_2=L$. To this end consider the following reduction
from \textsc{Partition} problem.

\medskip
\textsc{Partition}

\label{sec:1} {\bf Instance:} Nonnegative integers $w_1,\ldots,w_m$
such that $\sum_{k=1}^{m}w_k=2S$.

\medskip
{\bf Question:} Does there exist a subset $X\subseteq \{1,\ldots,
m\}$ such that $ \sum_{k\in X}w_k \\ =S$?
\medskip

Consider an instance $\mathcal{I}$ of {\sc Partition} and construct
an instance $\mathcal{I}'$ of $F2\;|\mbox{ exact }l_j\in
\{0,L\}|\;C_{\max}$.

Let $R>5 S$. Set $J=\{1,\ldots ,m+6\}$ and
\begin{eqnarray*}
   a_k&=&b_k=w_k, \; l_k=2R \mbox{ for } \quad k=1,\ldots m, \\
   a_{m+1} &=& b_{m+1}=R,\; l_{m+1}=0, \\
   a_{m+2} &=& b_{m+2}=R, \; l_{m+2}=2R, \\
   a_{m+3} &=& 0, \; b_{m+3}=R-S, \; l_{m+3}=0, \\
   a_{m+4} &=& R-S, \; b_{m+4}=0,\; l_{m+4}=0,\\
  a_{m+5} &=& 0, \; b_{m+5}=R, \;  l_{m+5}=0, \\
  a_{m+6} &=& R,\; b_{m+6}=0,\;  l_{m+6}=0.
\end{eqnarray*}
We will refer to the jobs in $\{1,\ldots, m\}$ as \emph{small} and
to the remaining six jobs as \emph{big}.

\begin{lemma}\ \label{l1}
\begin{description}
  \item[(i)] If $\sum_{k\in X}w_k =S$ for some subset $X\subseteq \{1,\ldots
m\}$, then there exists a feasible schedule $\sigma$ such that
$C_{\max}(\sigma)\leq 4R+4S$.
  \item[(ii)] If there exists a feasible schedule $\sigma$ such that
$C_{\max}(\sigma) \leq 4R+4S$, then $\sum_{k\in X}w_k =S$ for some
set $X\subseteq \{1,\ldots m\}$.
  \item[(iii)] If
$C_{\max}(\sigma) > 4R+4S$ for any feasible schedule $\sigma$, then
$C_{\max}^*\geq 5R-S$.
\end{description}
\end{lemma}
\begin{proof}

(i) Let $X\subseteq \{1,\ldots, m\}$ such that $ \sum_{k\in
X}w_k=S$. Then $ \sum_{k\in Y}w_k=S$ where $Y=\{1,\ldots,
m\}\setminus X$.

First of all point out that the whole construction presenting a
feasible schedule can be moved along the time line in both
directions. So the length of the schedule is the length of the time
interval between the starting time  of the first operation (which is
not necessarily equal to zero) and the completion time of the last
one.
\begin{figure}[h] \label{ex5}
\vspace{1mm}
\begin{center}
\includegraphics[scale=0.65]{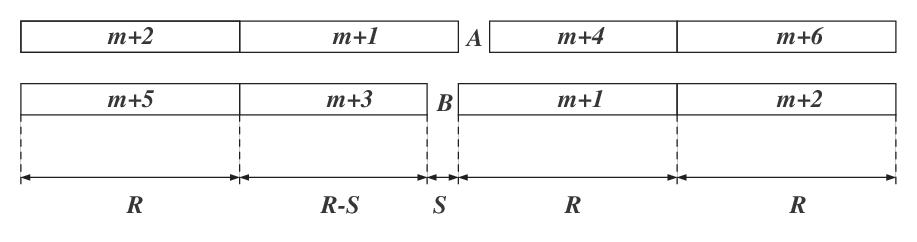}
  \caption{Scheduling the big jobs.}
\end{center}
\end{figure}
To construct the required schedule arrange the big jobs in the order
shown in Fig.~1. This construction has two idle intervals: $A$ on
machine~1 and $B$ on machine~2. The interval $A$ is between the end
of the first operation of job $m+1$ and the beginning of the first
operation of job $m+4$. The interval $B$ is between the end of the
second operation of job $m+3$ and the beginning of the second
operation of job $m+1$. Both intervals have length $S$.

For scheduling the small jobs we use the following rule. Schedule
the small jobs in $X$ in such a way that their first operations are
executed within the time interval $A$ in non-decreasing order of the
lengths. Correspondingly, w.l.o.g. we may assume that
$X=\{1,2,\ldots ,q\}$ and $w_1\leq w_2\leq \ldots \leq w_q$.

Denote by $A'$ the time interval between the end of the second
operation of job $m+2$ and the end of the last operation of job $q$.
It is easy to understand (see Fig.~2) that all the second operations
of jobs $\{1,\ldots , q\}$ fall within $A'$ and the length of $A'$
is equal to
$$\sum_{i=1}^q w_i+w_1
+(w_2-w_1)+(w_3-w_2)+\ldots+(w_q-w_{q-1})= \sum_{i=1}^q w_i+w_q,$$
which does not exceed $2S$.
\begin{figure}[h]
\vspace{1mm}
\begin{center}
\includegraphics[scale=0.6]{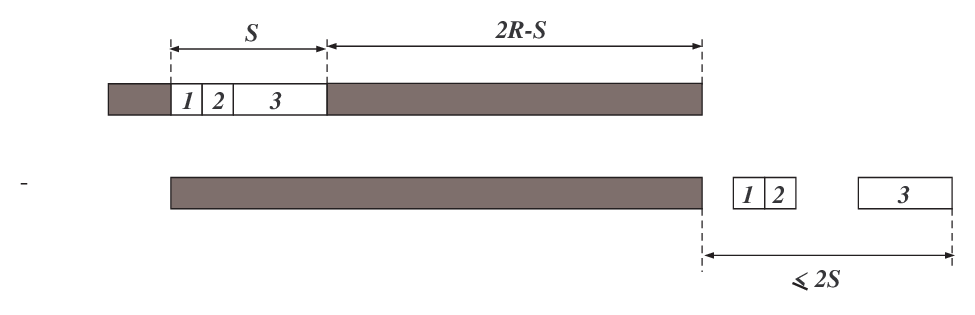}
  \caption{Scheduling the small jobs in $X$.}
\end{center}
\end{figure}
Now we observe that the construction is symmetric and schedule the
jobs in $Y$ quite similarly. More precisely, the second operations
of jobs in $Y$ are executed without interruption within the time
interval $B$ in non-increasing order of their lengths. Then the
first operations of these jobs fall within a time interval of length
at most $2S$.

Finally, we arrive at the schedule shown in Fig.~3. From the above
argument its length does not exceed $4R+4S$, as required.
\begin{figure}[h]
\vspace{1mm}
\begin{center}
\includegraphics[scale=0.6]{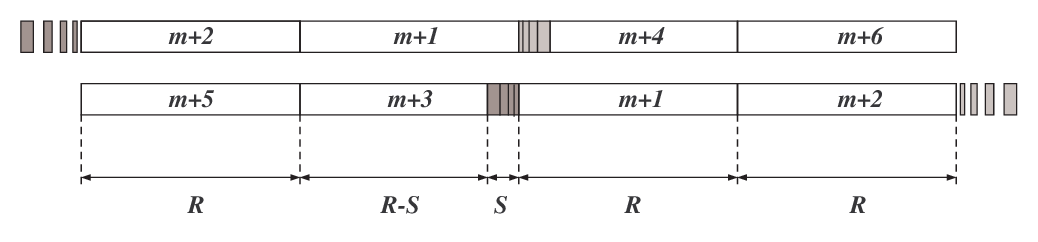}
  \caption{The jobs in $\{1,\ldots, m\}$ are executed
within the shaded intervals.}
\end{center}
\end{figure}

(ii) Let $\sigma$ be a feasible schedule with $C_{\max}(\sigma) \leq
4R+4S$. Observe first that in any schedule of length at most $4R+4S$
both operations of job $m+1$ are executed exactly within the lag
time interval of job $m+2$, since otherwise $C_{\max}(\sigma) \geq
5R$. So for these jobs we have the \emph{initial} construction shown
in Fig.~4. Denote by $t_0, t_1, t_2, t_3, t_4$ the junction times of
the operations of these jobs (see Fig.~4).

\begin{figure}[h]
\vspace{1mm}
\begin{center}
\includegraphics[scale=0.65]{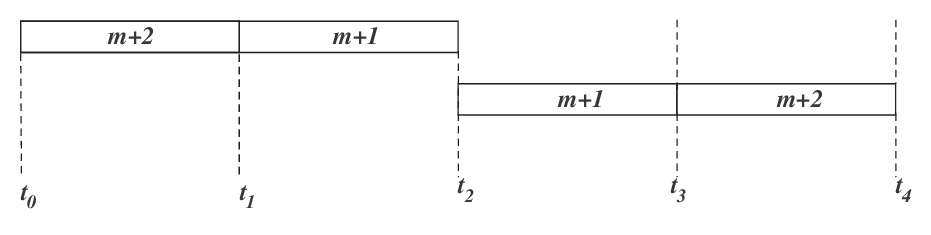}
  \caption{The  initial construction.}
\end{center}
\label{fig4}
\end{figure}
Observe that the schedule $\sigma$ has the following property (Q):
for any small job either it completes at time not earlier than
$t_1$, or starts at time not later than $t_3$. It follows from the
fact that otherwise the length of $\sigma$ is at least $5R$. Let $X$
be the subset of small jobs that complete at time not earlier than
$t_1$. Then $Y=\{1,\ldots , m\}\setminus X$ is the subset of small
jobs that start at time not later than $t_3$. For definiteness
assume that $X\not= \emptyset$. This immediately implies that job
$m+5$ starts at time not later than $t_0$. Then job $m+3$ starts
exactly at time $t_1$, since otherwise the length of $\sigma$ is at
least $5R-S$, which is greater than $4R+4S$ due to the choice of $R$
and $S$. Thus the second operations of the jobs in $X$ are executed
within the interval $[t_2-S,S]$. It follows that $Y\not=\emptyset$.
Moreover, a similar argument shows that the first operations of the
jobs in $Y$ are executed within the interval $[t_2, t_2+S]$. Thus we
have $\sum_{j\in X}{w_j}\leq S$ and $\sum_{j\in Y}{w_j}\leq S$,
which implies $\sum_{j\in X}{w_j}=\sum_{j\in Y}{w_j}=S$, as
required.

(iii) Let $\sigma$ be a feasible schedule. By the assumption of
(iii) $C_{\max}(\sigma)>4R+4S$. We may assume that $\sigma$ contains
the initial construction and satisfies property (Q) (see (ii)),
since otherwise we are done. Let $X$ and $Y$ be defined as in (ii).
From (i) it follows that $\sum_{j\in X}w_j\not=\sum_{j\in Y}w_j$.
W.l.o.g. we may assume that $\sum_{j\in X}w_j>\sum_{j\in Y}w_j$,
i.e., $\sum_{j\in X}w_j>S$. Then by property (Q) neither job $m+3$
nor job $m+5$ starts at time $t_1$. It follows that at least one of
these jobs is executed outside the interval $[t_0,t_4]$, which
implies $C_{\max}(\sigma)\geq 5R-S$ (one of the possible
configurations is shown in Fig.~5). 
\begin{figure}[h]
\vspace{1mm}
\begin{center}
\includegraphics[scale=0.60]{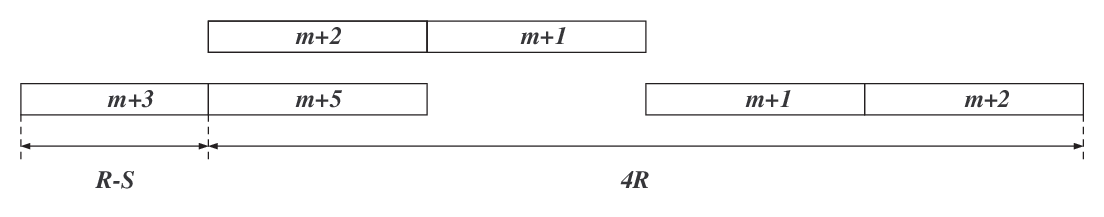}
  \caption{One of the possible configurations for (iii).}
\end{center}
\end{figure}
\end{proof}
Set $R=kS$. Then $5R-S=S(5k-1)$. On the other hand,
$4R+4S=4kS+4S=S(4k+4)$. The fraction
$$\frac{5k-1}{4k+4}$$
tends to $1.25$ as $k$ tends to infinity. Thus Lemma \ref{l1}
implies
\begin{theorem}
If the problem $F2\;|\mbox{ exact }l_j\in \{0,L\}\;|\;C_{\max}$
admits a $(1.25-\varepsilon)$-approximation algorithm, then $P=NP$.
\qed
\end{theorem}

\section{A $2$-approximation algorithm for \\ $F2\;|\mbox{ exact }l_j\in \{L_1,L_2\}\;|\;C_{\max}$}

In this section we present a simple $2$-approximation algorithm for
solving $F2\;|\mbox{ exact }l_j\in \{L_1,L_2\}\;|\;C_{\max}$.

We show first that the case when the delays are the same for all
jobs ($L_1=L_2=L$) is polynomially solvable. Note that any feasible
schedule $\sigma$ of an instance of $F2\;|\mbox{ exact
}l_j=L\;|\;C_{\max}$ can be associated with a feasible schedule
$\sigma'$ of the corresponding instance of $F2\;|\mbox{ exact
}l_j=0\;|\;C_{\max}$ and their lengths satisfy
$C_{\max}(\sigma)=C_{\max}(\sigma')+L$. More precisely, shifting the
second operations of all jobs to the left by distance $L$ gives a
feasible schedule to the problem with zero delays, and vise versa
(see Fig.~6). The problem $F2\;|\mbox{ exact }l_j=0\;|\;C_{\max}$
(all delays are equal to $0$) is nothing but the 2-machine no-wait
Flow Shop problem. The latter problem is known to be solvable in
$O(n\log n)$ time \cite{GG,GLS,BDDVW}. Therefore the problem
$F2\;|\mbox{ exact }l_j=L|\;C_{\max}$ is solvable in $O(n\log n)$
time for all $L\geq 0$.

\begin{figure}[t]
\vspace{1mm}
\begin{center}
\includegraphics[scale=0.75]{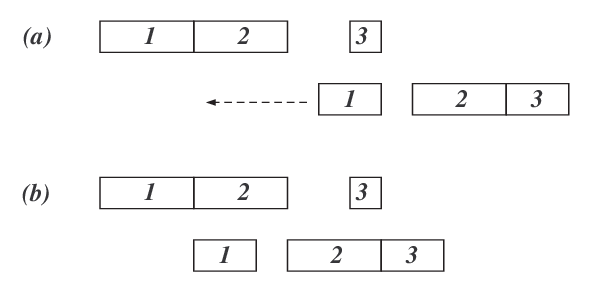}
  \caption{(a) A schedule with the same delay $L$ for all jobs; (b) the corresponding schedule
  with delay $0$ for all jobs. }
\end{center}
\end{figure}

Let $\mathcal{I}_1, \mathcal{I}_2$ be instances of $F2\;|\mbox{
exact }l_j|\;C_{\max}$ with disjoint sets of jobs $J_1$ and $J_2$.
Let $\sigma_k$, $k=1,2$, be feasible schedules of $\mathcal{I}_k$,
respectively. Consider an instance $\mathcal{I}$ of $F2\;|\mbox{
exact }l_j|\;C_{\max}$ formed by the union of $J_1$ and $J_2$.
Denote by $\sigma_1\oplus\sigma_2$ the schedule of $\mathcal{I}$
obtained from $\sigma_1$ and $\sigma_2$ by \emph{concatenation} of
schedules $\sigma_1$ and $\sigma_2$. More precisely, the schedule
$\sigma_1\oplus\sigma_2$ first executes the jobs in $J_1$ according
to the schedule $\sigma_1$ and then as earlier as possible starts
executing the jobs in $J_2$ according to the schedule $\sigma_2$. An
example with $J_1=\{(3,2,3), (2,2,4)\}$ and
$J_2=\{(4,0,2),(5,0,1)\}$ is depicted in Fig.~7.
\begin{figure}[h]
\vspace{1mm}
\begin{center}
\includegraphics[scale=0.65]{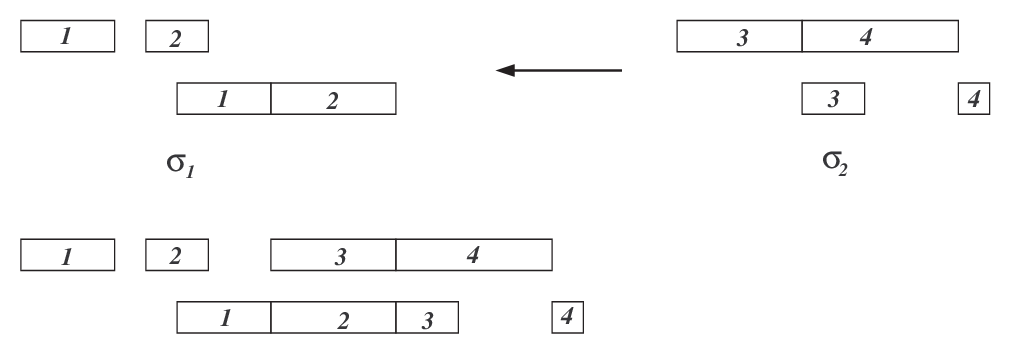}
  \caption{The schedule shown below is the concatenation of $\sigma_1$ and $\sigma_2$.}
\end{center}
\label{fig5}
\end{figure}

We now give a description of an approximation algorithm for the
problem $F2\;|\mbox{ exact }l_j\in \{L_1,L_2\}\;|\;C_{\max}$.

\medskip

\textbf{Algorithm} \textsc{Concatenation}

\medskip

Input: An instance $\{(a_j,l_j,b_j): j\in J\}$, $l_j\in
\{L_1,L_2\}$.

Output: A feasible schedule $\sigma$.

\medskip

1. Set $J_k=\{j\in J: l_j=L_k$, $k=1,2$\}. For $k=1,2$ form the
instances $\mathcal{I}_k=\{(a_j,L_k, b_j): j\in J_k\}$ of
$F2\;|\mbox{ exact }l_j=L\;|\;C_{\max}$.

2. Solve the instances $\mathcal{I}_k$, $k=1,2$. Let $\sigma_k$,
$k=1,2$, be optimal schedules of $\mathcal{I}_k$, respectively.

3. Set $\sigma=\sigma_1\oplus\sigma_2$.

\medskip

As mentioned above the time complexity of Step 2 is $O(n\log n)$. So
the overall running time of Algorithm \textsc{Concatenation} is
$O(n\log n)$.

The approximation bound is derived from the following easy lemmata.

\begin{lemma}
Let $C^*_{\max}$ be the length of an optimal schedule to the
instance $\{(a_j,l_j,b_j): j\in J\}$, $l_j\in \{L_1,L_2\}$. Then
$C_{\max}(\sigma_k)\leq C^*_{\max}$ for $k=1,2$.
\end{lemma}
\begin{proof} Evident. 
\end{proof}
\begin{lemma}
$C_{\max}(\sigma)\leq C_{\max}(\sigma_1)+C_{\max}(\sigma_2)$.
\end{lemma}
\begin{proof} Follows from the definition of operation $\oplus$.
\end{proof}
From Lemmata~2 and 3 we have
$$
C_{\max}(\sigma)\leq 2C^*_{\max}.
$$
Thus we arrive at the following
\begin{theorem}
Algorithm \textsc{Concatenation} runs in time $O(n\log n)$ and finds
a feasible schedule of $F2\;|\mbox{ exact }l_j\in
\{L_1,L_2\}\;|\;C_{\max}$ whose length is within a  factor of $2$ of
the optimum. \qed
\end{theorem}

\medskip

\noindent \textbf{Tightness}.

\smallskip

\noindent For $a>\delta >0$, consider the instance $\{(a_j,l_j,b_j):
j\in J\}$, $l_j\in \{L_1,L_2\}$,  consisting of $k$ jobs $(a,0,a)$
and a single job $(\delta, (k+1)a, \delta)$. It is clear that
$$C_{\max}^*=(k+1)a+2\delta$$
(see Fig. 8 for the optimal schedule). Let $\sigma_{conc}$ denote
the schedule returned by Algorithm \textsc{Concatenation}. Then
evidently
$$C_{\max}(\sigma_{conc})=ka+2\delta+(k+1)a.$$
\begin{figure}[h!]
\vspace{1mm}
\begin{center}
\includegraphics[scale=0.65]{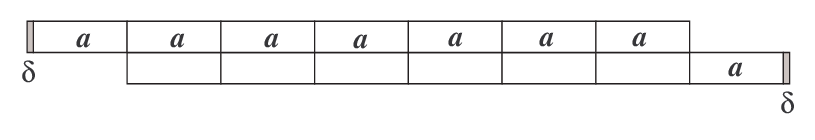}
  \caption{The optimal schedule.}
\end{center}
\end{figure}
Thus we have

$$\frac{C_{\max}(\sigma_{conc})}{C_{\max}^*}=
\frac{ka+2\delta+(k+1)a}{(k+1)a+2\delta}=1+\frac{ka}{(k+1)a+2\delta}$$
which tends to $2$ when $k$ tends to $\infty$. \qed


\begin{thebibliography}{}
%
%

\bibitem{AB}  A.A. Ageev and  A.E. Baburin,
Approximation algorithms for UET scheduling problems with exact
delays. Oper. Res. Letters 35(2007), 533--540.

\bibitem{AK} A.A. Ageev and A.V. Kononov, Approximation algorithms for
scheduling  problems with exact delays. In: Approximation and Online
Algorithms: 4th International Workshop (WAOA 2006), Zurich,
Switzerland, LNCS 4368 (2007), 1--14.

\bibitem{BDDVW} R.E. Burkard, V.G. Deineko, R. van Dal, J.A.A. van der Veen and G.J. Woeginger.
Well-solvable special cases of the traveling salesman problem: A
survey. SIAM Review, 40: 496--546, 1998.

\bibitem{Con} A. Condotta, Scheduling with due dates and time lags:
new theoretical results and applications. Ph.D. Thesis, 2011, The
University of Leeds, School of Computing, 156 pp.


\bibitem{ESS} M. Elshafei, H. D. Sherali, and J.C. Smith, Radar pulse
interleaving for multi-target tracking. Naval Res. Logist. 51
(2004), 79--94.


\bibitem{GG} P.C. Gilmore and R.E. Gomory.
Sequencing a one-state variable machine: a solvable case of the
traveling salesman problem. Operations Research 12 (1964), 655--679.

\bibitem{GLS} P.C.Gilmore, E.L. Lawler and D.B. Shmoys. Well solved cases, in E.L. Lawler,
J.K. Lenstra, A.H.G. Rinnooy Kan and D.B. Shmoys (Eds), The
Traveling Salesman Problem: A Guided Tour of Combinatorial
Optimization, Wiley, New York, pp. 87--143, 1986.


\bibitem{LLZ}J. Y.-T. Leung, H. Li, and H. Zhao,
Scheduling two-machine flow shops with exact delays. International
Journal of Foundations of Computer Science 18 (2007), 341--359.


\bibitem{SS} H.~D. Sherali and J.~C. Smith, Interleaving
two-phased jobs on a single machine, Discrete Optimization 2 (2005),
348--361.

\bibitem{Yu} W. Yu, The two-machine  shop problem with delays and the one-machine total
tardiness problem, Ph.D. thesis, Technische Universiteit Eindhoven,
1996.

\bibitem{YHL}W.~Yu, H.~Hoogeveen, and J.~K.~Lenstra,
Minimizing makespan in a two-machine flow shop with delays and
unit-time operations is NP-hard. J. Sched. 7 (2004), no. 5,
333--348.



\end{thebibliography}
\end{document}